\newcommand*{\SINGLECOL}{}
\ifdefined\SINGLECOL
\documentclass[12pt, draftclsnofoot, onecolumn]{IEEEtran}
\else
\documentclass[10pt]{IEEEtran}
\fi
\usepackage[noadjust]{cite}
\usepackage{amsmath}
\usepackage{amssymb}
\usepackage{amsthm}
\usepackage{bbm}
\usepackage{bm}
\usepackage{amsfonts}
\usepackage{booktabs}
\usepackage{comment}
\usepackage{graphicx}
\usepackage[capitalize]{cleveref}
\usepackage{caption}
\usepackage{subcaption}
\usepackage{algorithm}
\usepackage{algpseudocode}
\usepackage{balance}

\newtheorem{proposition}{Proposition}
\renewcommand{\epsilon}{\varepsilon}
\usepackage{xcolor}
\definecolor{plotcol1}{rgb}{.894,.102,.110}
\definecolor{plotcol2}{rgb}{.215,.492,.719}
\definecolor{plotcol3}{rgb}{.301,.684,.289}
\definecolor{plotcol4}{rgb}{.594,.305,0.637}
\definecolor{plotcol5}{rgb}{1,.498,0}
\definecolor{plotcol6}{rgb}{1,1,.2}
\DeclareMathOperator{\E}{\mathbb{E}}
\DeclareMathOperator{\Var}{\mathrm{Var}}

\begin{document}
\title{Common Message Acknowledgments: \\ Massive ARQ Protocols for Wireless Access}
\author{Anders E. Kal{\o}r,~\IEEEmembership{Graduate Student Member,~IEEE},
  Rados\l{}aw Kotaba,~\IEEEmembership{Member,~IEEE},
  and Petar Popovski,~\IEEEmembership{Fellow,~IEEE}
\thanks{The work has been supported by the Danish Council for Independent Research, Grant Nr. 8022-00284B SEMIOTIC, and by the Villum Investigator Grant ``WATER'' from the Velux Foundation, Denmark.}%
\thanks{The authors are with the Department of Electronic Systems, Aalborg University, Denmark (email: \{aek,rak,petarp\}@es.aau.dk).}}

\maketitle

\begin{abstract}
  Massive random access plays a central role in supporting the Internet of Things (IoT), where a subset of a large population of users simultaneously transmit small packets to a central base station. While there has been much research on the design of protocols for massive access in the uplink, the problem of providing message acknowledgments back to the users has been somewhat neglected. Reliable communication needs to rely on two-way communication for acknowledgement and retransmission. Nevertheless, because of the many possible subsets of active users, providing acknowledgments requires a significant amount of bits. Motivated by this, we define the problem of massive ARQ (Automatic Retransmission reQuest) protocol and introduce efficient methods for joint encoding of multiple acknowledgements in the downlink. The key idea towards reducing the number of bits used for massive acknowledgments is to allow for a small fraction of false positive acknowledgments. We analyze the implications of this approach and the impact of acknowledgment errors in scenarios with massive random access. Finally, we show that these savings can lead to a significant increase in the reliability when retransmissions are allowed since it allows the acknowledgment message to be transmitted more reliably using a much lower rate.
\end{abstract}


\section{Introduction}
A fundamental challenge in supporting the Internet of Things (IoT) is to enable grant-free, or uncoordinated, transmissions from a very large number of users~\cite{bockelmann16}. Furthermore, as the user activation is often triggered by physical phenomena, such as an event that generates sensory data, the traffic patterns are sporadic. Thus, at any instant, the resulting subset of active user that have something to transmit is random. This has initiated a large amount of research devoted to the design of random access schemes that can decode messages from a small random subset of users, often based on techniques derived from ALOHA~\cite{stefanovic13,paolini15} or compressed sensing~\cite{liu18,amalladinne20,fengler21}.

However, despite the great interest in transmission schemes for massive access, the problem of efficiently providing packet reception acknowledgments to a large number of users has been somewhat neglected. Yet, a message acknowledgment is often an useful signal for the application layer, and is necessary in order to implement retransmission schemes, referred to as Automatic Retransmission reQuest (ARQ). By allowing the transmission to terminate as soon as the message has been decoded, ARQ provides a mechanism for implicit rate adaptation. This can greatly increase the transmission reliability, especially in time-varying channels, such as those where the transmissions are subject to fading and interference from other users as is often the case in massive random access~\cite{caire01arq,wu10arq}. This has been exploited in several random access schemes, e.g., those inspired by rateless codes~\cite{stefanovic13,shirvanimoghaddam17}, which rely on message acknowledgments in order to work. Although these schemes require only a single bit of common feedback, it can be beneficial in practice to provide early feedback as soon any individual user is decoded in order to minimize the interference from imperfect SIC. In essence, our work treats the problem of massive ARQ and thus expands the problem space of the area of massive wireless access.

Compared to grant-based access scenarios, where the BS can send an acknowledgment to a user using a single bit (ACK/NACK), acknowledging a set of users decoded from a grant-free access scenario requires the BS to encode the user identities or some other information that can be used to identify the users that it wants to acknowledge. Encoding the user identifiers requires a significant number of bits when the number of users is large. A na\"ive attempt to encode acknowledgments to $K$ users out of a total of $N$ users could be to simply concatenate the identifiers of the $K$ users and transmit an acknowledgment packet of $K\log_2(N)$ bits. However, this approach has two significant drawbacks. First, it requires a variable-length packet, which may not be desired in many protocols that rely on time-division multiplexing. Second, as we will show, it is possible to significantly reduce the number of bits required to encode the acknowledgments by applying source coding techniques to jointly encode the acknowledgments for all $K$ users.

In order to achieve substantial reductions in the acknowledgment message length, our key proposal is to allow for a small but non-negligible fraction of \emph{false positive} acknowledgments, i.e., that a transmitting user erroneously determines that its message is among the acknowledged messages. Such errors are atypical in existing systems, which are often designed to suppress false positives using error detection mechanisms such as cyclic redundancy checks (CRCs), or by encoding the feedback message such that false positives are very rare at the cost of a larger \emph{false negative} probability~\cite{ostman21}. The reason for this is that false positive acknowledgments remain undetected after a transmission round and thus can be hard to resolve and lead to unreliable communication. This is in contrast to false negative errors, which may for instance occur if there are errors in the CRC but the message is intact, and for which the cost is merely an unnecessary retransmission. In this sense, a false positive acknowledgment can be ``fatal'' as it leads to the situation where the user believes that its message was successfully received by the BS when it in fact was lost. Despite this, it turns out that introducing false positives can somewhat surprisingly increase the reliability in the massive access scenario with ARQ. Nevertheless, in cases where false positive acknowledgments cannot be tolerated, we note that they can be detected and subsequently resolved using mechanisms at higher layers such as packet numbering, or by piggybacking an acknowledgment bit with the next downlink message at the cost of a detection delay.

\subsection{Background and Related Work}
The idea of joint feedback encoding in the massive random access scenario has been exploited in~\cite{kang21} to design feedback for collision-free scheduling of the users that are active in the uplink. However, they assume that the uplink was error-free, which makes the use of acknowledgments obsolete in the first place.
Joint encoding of acknowledgments has been studied in~\cite{yang16csack}, where a method based on compressed sensing is proposed. The overall idea is that base station constructs a sparse binary vector in $\mathbb{R}^N$ where the non-zero entries represent the decoded users, and projects it to a lower dimension using a sensing matrix. The compressed representation is then transmitted over the channel and the users try to recover the sparse vector of decoded users using a variant of approximate message passing designed to suppress false negatives. While this method shares many similarities with the ones that we study, the main disadvantage is that decoding algorithms are computationally hard, if not infeasible, for the values of $N$ that we consider, e.g., $N=2^{32}$ or $N=2^{64}$. Joint acknowledgments have also been used in the S-band Mobile Interactive Multimedia (S-MIM) satellite system~\cite{scalise13smim}, where acknowledgment messages can be constructed by concatenating the CRCs of the uplink messages, thereby introducing false positive acknowledgments. However, there is no rigorous analysis of the scheme and, as we will show, the approach is sub-optimal in terms of the fraction of false positives produced for a given message length.

The trade-off between false positives and false negatives acknowledgments has been studied thoroughly for ARQ and hybrid ARQ (HARQ) in the single-user setting, where only a single-bit acknowledgment message is needed. The general conclusion from these studies is that the probability of false positive acknowledgments needs to be significantly smaller than the uplink error probability, since they, contrary to uplink failures and false negative acknowledgments, cannot be repaired by a retransmission~\cite{draper08,wu11}. The same result holds in the finite blocklength regime, where the downlink message should be designed to achieve low false positive probability, but the false negative probability should be held constant and relatively large independently of the total reliability requirement~\cite{ostman21}. Although the reliability of the feedback is generally less important when the maximum number of transmissions is small since the uplink reliability plays a more significant role in determining the total reliability, these results hold even with as few as two transmission rounds~\cite{wu11}. Nevertheless, because of the large feedback message required in massive access regime and the fact that the feedback acknowledges multiple users, these results cannot be directly transferred to the scenario that we consider.

Finally, we note that the practical methods that we will consider rely extensively on the use of hash functions, which have also been widely used in the design of uplink schemes for random access. In~\cite{starzetz2009hashing} hashing is used to resolve collisions, \cite{pratas16bloom,zhang18bloom,luo20bloom} consider various applications of Bloom filters, which rely heavily on hashing, and in~\cite{vem19ura} hashing is used to decorrelate interference based on the users' messages in unsourced random access.

\subsection{Contributions and Paper Organization}
The paper has three main contributions. \emph{First,} it is the core idea of allowing false positives. We show that by allowing a small fraction of false positive acknowledgments, the number of bits required for the feedback message can be significantly reduced, while the introduction of false negative acknowledgment does not yield comparable savings. Furthermore, we present various practical methods for efficient encoding of acknowledgments with false positives. \emph{Second,} we study how the distribution of the number of active users impacts the feedback message, and derive closed-form bounds on the false positive probability based on the first and second moments of the distribution. \emph{Third,} we quantify the impact of false positive acknowledgments on the overall reliability by studying transmission schemes with multiple transmission rounds. In this context, we show that the message length reduction that results from introducing false positives allows the feedback to be transmitted with a much lower rate, which in turn results in a significant increase in the overall reliability.

We note that, in both grant-free and grant-based settings and irrespectively of the feedback encoding, feedback can be designed either in an adaptive or non-adaptive manner. Adaptive feedback schemes are intrinsically non-trivial due to the half-duplex structure of most wireless systems, which requires the feedback instants to be either fully pre-planned or controlled by the transmitting user. In this paper, the focus is on the feedback message and assume that the feedback moments are known.

The remainder of the paper is organized as follows. \cref{sec:sysmodel} introduces the overall system model. Information-theoretic bounds for a fixed number of decoded users are presented in \cref{sec:inftheorybounds}, and \cref{sec:practicalschemes} introduces and analyzes a number of practical encoding schemes for this setting. \Cref{sec:random_activations} analyzes the case in which the number of decoded users is random, and the case with multiple transmission rounds is analyzed in \cref{sec:feedback_retx}. Finally, numerical results are presented in \cref{sec:num_res} and the paper is concluded in \cref{sec:conclusion}.

\section{System Model}\label{sec:sysmodel}
We consider a typical massive access scenario comprising a single base station (BS) that serves a massive set of potentially active users $[N]=\{1,2,\ldots,N\}$ (typically $N$ is in the order of thousands). As is often the case in practical systems, we assume that each of the $N$ users has a unique identifier known to both the users and the BS. If the BS requires an initial handshake procedure for users to join the network, then $N$ corresponds to the number of users associated with the BS, and $N$ will be in the order of thousands (for instance in NB-IoT, the Cell Radio Network Temporary Identifier (C-RNTI) can identify up to $N=65523$ users~\cite{ts36321}). On the other hand, if no such procedure exists, then each user can have a globally unique identifier, such as a MAC address, and $N$ will be in the order of $2^{32}$ to $2^{64}$.

We assume a general frame structure in which the air interface is divided into a number of recurring random access opportunities in which a random subset $\mathcal{A}\subseteq [N]$ of users are active and transmit their messages in the uplink. The uplink transmission is followed by a downlink feedback message, multicasted by the BS, that provides acknowledgments to the users that the BS decoded in the uplink. Users that receive an acknowledgment have completed their transmission, while users that do not receive an acknowledgment are allowed to retransmit up to $L-1$ times. We assume that each uplink message contains the transmitter's identifier such that the BS is able to determine the identity of the sender upon decoding of the packet\footnote{We make this assumption for clarity of presentation, but the analysis holds even if there is no identity (e.g., as in unsourced random access~\cite{polyanskiy17}) by treating the messages as temporary identities. In that case $N$ corresponds to the number of distinct messages.}.
In general, the number of active users is random and typically will be much smaller than $N$. The set of active users (including its cardinality) is unknown to the BS, which tries to recover it from the received signals. We denote the set of recovered users by $\mathcal{S}=\{s_1,s_2,\ldots,s_K\}$, where $s_k\in[N]$ and assume that, conditioned on $K$, $\mathcal{S}$ is drawn uniformly from the set of all $K$-element subsets of $[N]$, denoted $[N]_K=\{\mathcal{K}\subseteq [N] \mid |\mathcal{K}|=K\}$. 
Due to decoding errors, $\mathcal{S}$ may be different from the actual set of active users $\mathcal{A}$. We denote by $\epsilon_{\mathrm{ul},n}$ the probability that a transmitting user $n$ is not decoded. This probability typically depends on the random access mechanisms as well as the value of $K$, the signal-to-noise ratios (SNRs) of the transmitting users, etc.

To make the transmitters aware of potential errors and to ensure reliable transmission, the BS transmits a common $B$-bit feedback message after the random access opportunity that allows the users to determine whether their own identifier is a member of $\mathcal{S}$. The message is transmitted through a packet erasure channel so that the packet is received by user $n$ with probability $1-\epsilon_{\mathrm{dl},n}$. The erasure probability depends on the SNR of the individual users, the channel, and the transmission rate of the feedback\footnote{In  practice, an erasure channel represents the case where the decoder can detect if a packet is decoded incorrectly, e.g., through an error-detecting code. The size of such a code with negligible false positive probability is small compared to the size of the feedback message, and thus we will ignore the overhead it introduces.}. Formally, such a feedback scheme is defined by an encoder, the downlink channel, and a set of decoders, one for each user. We define the encoder as
\begin{equation}
  f: [N]_K \to \{0,1\}^B,
\end{equation}
and the erasure channel as
\begin{align}
  \Pr(Y_n=X\mid X\in \{0,1\}^B)&=1-\epsilon_{\mathrm{dl},n},\\
  \Pr(Y_n=\mathrm{e}\mid X\in \{0,1\}^B)&=\epsilon_{\mathrm{dl},n},
\end{align}
where $X$ is the packet transmitted by the BS, $Y_n$ is the packet received by user $n$, and $\mathrm{e}$ denotes an erasure. Finally the individual decoders are defined as
\begin{equation}
  g_n: \{0,1\}^B \cup \mathrm{e} \to \{0,1\},\qquad n=1,\ldots,N
\end{equation}
which output $1$ if user $n$ is believed to be a member of $\mathcal{S}$ and $0$ otherwise (throughout the paper we will assume that the decoder outputs $0$ if it observes an erasure). Both the encoder and the decoders may depend on $K$ (which is random), but this dependency can be circumvented by encoding $K$ separately in the feedback message at an average cost of approximately $H(K)$ bits where $H(\cdot)$ is the entropy function\footnote{A pragmatic alternative when the activation distribution is unknown would be to assume that at most $K'$ users can be decoded simultaneously and then dedicate a fixed number of $\log_2(K')$ bits to describe $K$.}. As we will see, this overhead is minimal when compared to the number of bits required to encode the acknowledgments in most settings of practical interest. We will refer to $B$ as the message length of a scheme.

To characterize the performance of a feedback scheme, we define the false positive (FP) probability, denoted $\epsilon_{\mathrm{fp}}$, as the probability that a user whose uplink message was not decoded $n\in \mathcal{A}\setminus \mathcal{S}$ erroneously concludes that it belongs to $\mathcal{S}$
\begin{align}
  \epsilon_{\mathrm{fp}}&=\E\left[\Pr\left(g_n\left(f\left(\mathcal{S}\right)\right)=1 \mid n\in \mathcal{A}\setminus \mathcal{S}\right) \mid K\right],
\end{align}
where the expectation is taken over $n$ and the distribution $p(\mathcal{S}|K)$ (but not the channel, which we will treat independently). Similarly, we define the false negative (FN) probability $\epsilon_{\mathrm{fn}}$ as the probability that a decoded user $n\in\mathcal{S}$ incorrectly concludes that it does not belong to the set
\begin{align}
  \epsilon_{\mathrm{fn}}&=\E\left[\Pr\left(g_n\left(f\left(\mathcal{S}\right)\right)=0 \mid n\in \mathcal{S}\right) \mid K\right].
\end{align}
Note that these definitions ignore the channel, and thus allow us to treat $\epsilon_{\mathrm{fp}}$ and $\epsilon_{\mathrm{fn}}$ independently of the event of an erasure. Note also that both $\epsilon_{\mathrm{tp}}$ and $\epsilon_{\mathrm{fp}}$ are conditioned on $K$. We discuss the case when $K$ is random further in \cref{sec:random_activations}.

Using these definitions, we denote by $B^*$ the minimum message length $B$ required for a scheme with $K$ active users out of $N$ that achieves false positive and false negative probabilities at most $\epsilon_{\mathrm{fp}}$ and $\epsilon_{\mathrm{fn}}$, respectively.

\section{Information Theoretic Bounds}\label{sec:inftheorybounds}
We first consider the source coding part of the problem, namely the functions $f$ and $g_n$ defined previously, while for clarity ignoring the erasure channel between the BS and the users. Specifically, in this section we derive information theoretic bounds on the minimum message length $B$ required for the feedback message. To start with, we treat $K$ as constant, and thus neglect the bits required to encode $K$ in the message, which would be the same for all schemes.

\subsection{Error-Free Coding}
We first consider error-free schemes, i.e., schemes that have $\epsilon_{\mathrm{fp}}=\epsilon_{\mathrm{fn}}=0$. A na\"ive construction of the feedback message is to concatenate the $K$ identifiers in $\mathcal{S}$ to produce a message of $B=K\log_2(N)$ bits. However, such a construction is sub-optimal because there are only $\binom{N}{K}$ subsets of $K$ users, and $\log_2\binom{N}{K}$ bits are sufficient to distinguish each subset. This leads to the feedback message length
\begin{align}
  B_{\text{error-free}}^* &= \left\lceil\log_2\binom{N}{K}\right\rceil\label{eq:bound_err_free}\\
  &\ge \left\lceil K\log_2(N/K)\right\rceil,\label{eq:lower_bound_err_free}
\end{align}
where the inequality follows from $\binom{N}{K}\ge (N/K)^K$. This bound can be achieved efficiently in practice using e.g., enumerative source coding~\cite{schalkwijk72,cover73}.

\subsection{Encoding with Bounded Errors}
The required feedback message length for the error-free encoding scales with the logarithm of $N$, which can be significant when $N$ is in the order of $2^{32}$ or $2^{64}$. One way to reduce the impact of $N$ is to allow for non-zero false positive and false negative probabilities. To do so, a feedback message must acknowledge at most $K+\epsilon_{\mathrm{fp}}N$ users, out of which at least $(1-\epsilon_{\mathrm{fn}})K$ must be in $\mathcal{S}$. For $\epsilon_{\mathrm{fp}}<0.5$ (which is typically the region of interest), it can be shown using combinatorial arguments that~\cite{pagh01} (see \cref{app:inf_lowerbound} for details)
\begin{align}
  \begin{split}
    B_{\text{fp,fn}}^* &\ge \log_2\binom{N}{K}\\
    &\qquad-\log_2\left(K\binom{\lfloor\epsilon_{\mathrm{fp}}N\rfloor + K}{\lceil(1-\epsilon_{\mathrm{fn}})K\rceil}\binom{N}{\lfloor\epsilon_{\mathrm{fn}}K\rfloor}\right)
  \end{split}\label{eq:lower}\\
  \begin{split}
    &\ge K\log_2\left(\frac{1}{\epsilon_{\mathrm{fp}}+\frac{K}{N}}\right) -
    K\log_2\left(\frac{e}{1-\epsilon_{\mathrm{fn}}}\right)\\
    &\quad - \epsilon_{\mathrm{fn}}K\log_2\left(\frac{1-\epsilon_{\mathrm{fn}}}{\epsilon_{\mathrm{fn}}\left(\epsilon_{\mathrm{fp}}+\frac{K}{N}\right)}\right) - \log_2(K),
  \end{split}\label{eq:lower2}
  \end{align}
where \eqref{eq:lower2} follows from the observation that rounding cannot decrease the message length and the inequality $\left(\frac{N}{K}\right)^K\le \binom{N}{K}\le (eN/K)^K$. Note that if $K$ is held constant, the bound is independent of $N$ as $N\to\infty$.

The introduction of false positives has the potential to offer significantly greater gains than false negatives. In particular, when $\epsilon_{\mathrm{fn}}$ is small as is typically desired, the required message length is only negligibly smaller than the one required if no false negatives were allowed. The reason for this is that the set of potential false negatives, $\mathcal{S}$, is much smaller than the set of potential false positives $[N]\setminus\mathcal{S}$. When $\epsilon_{\mathrm{fn}}=0$, the bound can be tightened further as~\cite{carter78,dietzfelbinger08}
\begin{align}
  B_{\text{fp}}^* &\ge K\log_2\left(1/\epsilon_{\mathrm{fp}}\right) - \frac{\log_2(e)(1-\epsilon_{\mathrm{fp}})K^2}{\epsilon_{\mathrm{fp}}N + (1-\epsilon_{\mathrm{fp}})K},\label{eq:lower_nofn}
\end{align}
where the last term vanishes as $N\to\infty$ for fixed $K$.

A (non-constructive) achievability bound for the case with $\epsilon_{\mathrm{fn}}=0$ and $K\le N\epsilon_{\mathrm{fn}}$ was provided in~\cite{carter78}. The overall idea is to sequentially generate all $\lfloor N\epsilon_{\mathrm{fp}}\rfloor$-element subsets of $[N]$, and then transmit the index of the first subset that includes all $K$ elements in $\mathcal{S}$. Using a set-cover theorem by Erd\H{o}s and Spencer~\cite[Theorem 13.4]{erdos74}, they show that the number of $\lfloor N\epsilon_{\mathrm{fp}}\rfloor$-element subsets required to cover all $K$-element subsets of $[N]$, denoted $M(N,\lfloor N\epsilon_{\mathrm{fp}}\rfloor,K)$, is upper bounded by
\begin{align}
  M(N,\lfloor N\epsilon_{\mathrm{fp}}\rfloor,K) \le \left(1+\ln\binom{\lfloor N\epsilon_{\mathrm{fp}}\rfloor}{K}\right)\frac{\binom{N}{K}}{\binom{\lfloor N\epsilon_{\mathrm{fp}}\rfloor}{K}}.
\end{align}
Taking the logarithm and bounding the binomial coefficients gives the following upper bound on the required feedback message length
\begin{align}
  \begin{split}
    B_{\text{fp}}^* &\le \log_2\binom{N}{K}-\log_2\binom{\lfloor N\epsilon_{\mathrm{fp}}\rfloor}{K}\\
    &\qquad + \log_2\left(1+\ln\binom{\lfloor N\epsilon_{\mathrm{fp}}\rfloor }{K}\right)
  \end{split}\label{eq:upper_nofn1}\\
  \begin{split}
    &\le K\log_2\left(e/\epsilon_{\mathrm{fp}}\right)
+ \log_2\left(1+K\ln\left(\frac{N\epsilon_{\mathrm{fp}}}{K}\right)\right).
  \end{split}\label{eq:upper_nofn2}
\end{align}
Note that this also serves as an upper bound for the case with $\epsilon_{\mathrm{fn}}>0$. By comparing \eqref{eq:upper_nofn2} to the lower bound in \eqref{eq:lower_nofn}, it can be seen that the bounds are tight within an additive term $O(\log\log N)$ as $N\to\infty$ and $K$ is held constant, i.e., for sufficiently large $N$,
\begin{align}
  B_{\text{fp}}^*=K\log_2(1/\epsilon_{\mathrm{fp}})+ O(\log\log N),\label{eq:asymptotic_bound}
\end{align}
which is lower than the error-free scheme in \cref{eq:lower_bound_err_free} when $\epsilon_{\mathrm{fp}}\ge K/N$. To illustrate the potential gain of introducing a small fraction of false positives, suppose $N=2^{32}$ and $K=100$. Encoding the acknowledgment in an error-free manner requires approximately $B=\log_2\binom{2^{32}}{100}\approx 2675$ bits, while only $B=100\log_2(100)\approx 664$ bits are required if we can tolerate $\epsilon_{\mathrm{fp}}=0.01$, and $B=100\log_2(10000)\approx 1329$ bits for $\epsilon_{\mathrm{fp}}=0.0001$.

\section{Practical Schemes}\label{sec:practicalschemes}
In this section, we analyze a number of practical designs of $f$ and $g_n$, and compare them to the bounds derived in the previous section. Motivated by the fact that false negatives provide little reduction in the feedback message length, we will restrict ourselves to schemes with $\epsilon_{\mathrm{fn}}=0$. Furthermore, unless otherwise noted, we will again assume that the number of decoded users $K$ is fixed and defer the discussion of random activations to \cref{sec:random_activations}.

\subsection{Optimal Scheme based on Linear Equations}\label{sec:le}
We consider construction based on solving a set of linear equations in a Galois field, first proposed in~\cite{dietzfelbinger08,porat2009optimal}. To simplify the analysis, we will assume that we have access to both a fully random hash function, $h_1$, and a universal hash function, $h_2$. An $(n,b)$-family of fully random hash functions is a family of functions $h_1:[n]\to [b]$ such that for each value $x\in[n]$, it outputs a value chosen uniformly at random from $[b]$. Similarly, an $(n,b)$-family of universal hash functions is a family of functions $h_2:[n]\to [b]$ such that for a hash function $h_2$ chosen uniformly at random and for any two distinct values $x,y\in[n]$, $\Pr(h_2(x)=h_2(y))\le 1/b$. While fully random  hash functions have desirable properties, they are not practical as they require an exponential number of bits to store. Nevertheless, in many practical problems the fully random hash function can be replaced by a simpler hash function with a negligible penalty, especially when the input is randomized~\cite{chung_mitzenmacher_vadhan_2021}. On the other hand, universal hash functions can be implemented efficiently in practice. Regardless of the type of hash function, the event that $h(x)=h(y)$ is typically referred to as a \emph{collision}. 

Returning to the encoding scheme, suppose we have a fully random hash function $h_1:[N]\to \mathrm{GF}(2^{\lceil\log_2(1/\epsilon_{\mathrm{fp}})\rceil})^K$, i.e., mapping from $[N]$ to $K$-element vectors in $\mathrm{GF}\left(2^{\lceil\log_2(1/\epsilon_{\mathrm{fp}})\rceil}\right)$, and a universal hash function $h_2: [N]\to 2^{\lceil\log_2(1/\epsilon_{\mathrm{fp}})\rceil}$. Then, we can construct the equation $h_1(s_k)\cdot z=h_2(s_k)$ in $\mathrm{GF}(2^{\lceil\log_2(1/\epsilon_{\mathrm{fp}})\rceil})$, where $\cdot$ is the inner product. By constructing an equation for each $s_k\in\mathcal{S}$, we obtain the set of $K$ equations with $K$ variables $H_1z=h_2$, where $H_1\in\mathrm{GF}(2^{\lceil\log_2(1/\epsilon_{\mathrm{fp}})\rceil})^{K\times K}$ is the matrix of rows $h_1(s_1),\ldots,h_1(s_K)$ and $h_2\in\mathrm{GF}(2^{\lceil\log_2(1/\epsilon_{\mathrm{fp}})\rceil})^{K}$ is the vector with elements $h_2(s_1),\ldots,h_2(s_K)$. In order for this system to have a solution, we require $H_1$ to be full rank. It can be shown that this happens with probability at least $1-\frac{1}{2^{\lceil\log_2(1/\epsilon_{\mathrm{fp}})\rceil}-1}$~\cite{porat2009optimal}, which is large for the values of $\epsilon_{\mathrm{fp}}$ that we consider (e.g., greater than $0.99$ for $\epsilon_{\mathrm{fp}}=0.01$ and greater than $0.9999$ for $\epsilon_{\mathrm{fp}}=0.0001$). By repeating the procedure with new hash functions $h_1', h_1'',\ldots$, the probability of generating a matrix with full rank can be made arbitrarily large at the cost of a message length penalty required to store the number of trials. In practice, this penalty is negligible compared to the total size of the message. For instance, with $\epsilon_{\mathrm{fp}}=0.01$ and up to 16 trials, requiring only four additional bits, the failure probability is in the order of $10^{-34}$.

Provided that the resulting matrix $H_1$ has full rank, we can obtain the solution $z$ to the set of equations. A decoder can then check whether an identifier $n$ is contained in the set by simply checking if $h_1(s_k)\cdot z=h_2(s_k)$. Thus, neglecting the potential overhead caused by repeating the hashing procedure, only the vector $z$ needs to be communicated, which contains $K$ entries of $\lceil\log_2(1/\epsilon_{\mathrm{fp}})\rceil$ bits each. Combining these observations, we obtain the feedback message length
\begin{align}
  B_{\mathrm{le}} = K\lceil\log_2(1/\epsilon_{\mathrm{fp}})\rceil,\label{eq:r_le}
\end{align}
which, disregarding the rounding, matches the asymptotic information theoretic bound in \cref{eq:asymptotic_bound}. As we will see in \cref{sec:num_res}, the practical performance matches closely with the bound.

It is worth noting that finding $z$ uses Gaussian elimination, which requires $O(K^3)$ operations. This makes the method infeasible for large $K$. However, the operation can be performed fast as long as $K$ is at most in the order of hundreds, which is the main interest in this paper. When $K$ is larger, the construction can be improved by introducing sparsity in $H_1$ at the cost of a small overhead, see e.g.,~\cite{porat2009optimal,dietzfelbinger08,BuRR2021}. Alternatively, one of the schemes presented in the following can be applied.

\subsection{Bloom Filter}
A Bloom filter~\cite{bloom70} uses $T$ independent universal hash functions $h_i: [N]\to [B]$ for $i=1,\ldots,T$, and is constructed by setting the message bits at positions $\{h_i(s_k) \mid s_k\in\mathcal{S}, i=1,\ldots,T \}$ equal to '1' and the remaining bits equal to '0'. In order to decode the message and check whether an identifier $n$ belongs to the set, the decoder simply checks if the bits at positions $\{h_i(n) \mid i=1,\ldots,T \}$ are equal to '1'. Clearly, the decoder can only observe false positives and not false negatives.

It can be shown that the minimum false positive probability is obtained when the probability that a given bit is '1' is exactly $1/2$, and that $T$ should be chosen as $T=(B/K)\ln(2)$ to achieve this~\cite{broder04} (in practice, one needs to round to the nearest integer). The resulting false positive probability is non-trivial, but can be approximated as~\cite{broder04}
\begin{align}
  \epsilon_{\mathrm{fp}}\approx 2^{-\lceil(B/K)\ln(2)+0.5\rceil}.
\end{align}
By assuming equality in the approximation we obtain
\begin{equation}
  B_{\mathrm{bf}} = K\log_2(e)\log_2(1/\epsilon_{\mathrm{fp}}),\label{eq:r_bf}
\end{equation}
revealing that Bloom filter is approximately within a factor $\log_2(e)\approx 1.44$ of the asymptotic lower bound in \cref{eq:asymptotic_bound}. Nevertheless, contrary to the scheme based on linear equations, a Bloom filter constructed to be optimal for a certain value of $K$ can be used with any other value of $K$ (although sub-optimally).

\subsection{Hash Concatenation}
As a final scheme, we consider the concatenation of hashes of each decoded identifier, which is practically equivalent to the one used in S-MIM~\cite{scalise13smim} (where CRCs are used as hash functions). Specifically, we consider a universal hash function $h:[n]\to [2^b]$, so that the message constructed by concatenating the hashes of each of the $K$ users has a length of $Kb$ bits. The probability that the hash of an arbitrary user that is not among the $K$ decoded users collides with any of the decoded users is
\begin{align}
\epsilon_{\mathrm{fp}}&=1-\left(1-\frac{1}{2^b}\right)^K.
\end{align}
By rearranging and ceiling to ensure $b$ is integer we obtain $b=\left\lceil-\log_2\left(1-(1-\epsilon_{\mathrm{fp}})^{1/K}\right)\right\rceil$. The feedback message length is then bounded by
\begin{align}
  B_{\mathrm{hc}}&=K\left\lceil-\log_2\left(1-(1-\epsilon_{\mathrm{fp}})^{\frac{1}{K}}\right)\right\rceil\\
  &\ge K\left\lceil-\log_2\left(1-e^{-\frac{\epsilon_{\mathrm{fp}}}{K(1-\epsilon_{\mathrm{fp}})}}\right)\right\rceil\\
  &\ge K\left\lceil-\log_2\left(\frac{\epsilon_{\mathrm{fp}}}{K(1-\epsilon_{\mathrm{fp}})}\right)\right\rceil\\
  &=K\left\lceil\log_2\left(1/\epsilon_{\mathrm{fp}}\right) + \log_2\left(K(1-\epsilon_{\mathrm{fp}})\right)\right\rceil,\label{eq:r_trunc}
\end{align}
where the first inequality follows from $1-x\ge e^{-\frac{x}{1-x}}$ for $0\le x<1$ and that $-\log_2(1-x)$ is monotonically increasing for $x<1$, and the second inequality is due to $1-e^{-x} \le x$ for $x\ge 0$ and that $-\log_2(x)$ is monotonically decreasing.
The last term in \cref{eq:r_trunc} is strictly positive when $K>\frac{1}{1-\epsilon_{\mathrm{fp}}}$, which is the case for the values of $K$ and $\epsilon_{\mathrm{fp}}$ that we are interested in. Thus, the scheme requires approximately $K\log_2(K(1-\epsilon_{\mathrm{fp}}))$ bits more than the lower bound in \cref{eq:asymptotic_bound}. Note that when the identifiers of the active users are drawn uniformly at random, the hash function can be replaced by a simple truncation of the identifiers to $b$ bits. Furthermore, as with the Bloom filter, this scheme can be used sub-optimally without knowing the instantaneous value of $K$, as long as the number of active users is less than $B_{\mathrm{hc}}/b$.

\subsection{Comparison}
\Cref{fig:bounds} compares the feedback message lengths for the practical schemes to the information theoretic bounds for $\epsilon_{\mathrm{fp}}=0.01$ and $\epsilon_{\mathrm{fp}}=0.0001$ when $N=2^{32}$. As expected, the upper (UB) and lower (LB) information theoretic bounds are very tight (within 14 bits in the considered range), and the scheme based on linear equations performs very close to these bounds. On the other hand, both the Bloom filter and hash concatenation require significantly more bits, but are still better than the error-free schemes, and have the advantage that they do not require knowledge of the instantaneous value of $K$ in order to be decoded.

\begin{figure}
  \centering
  \includegraphics{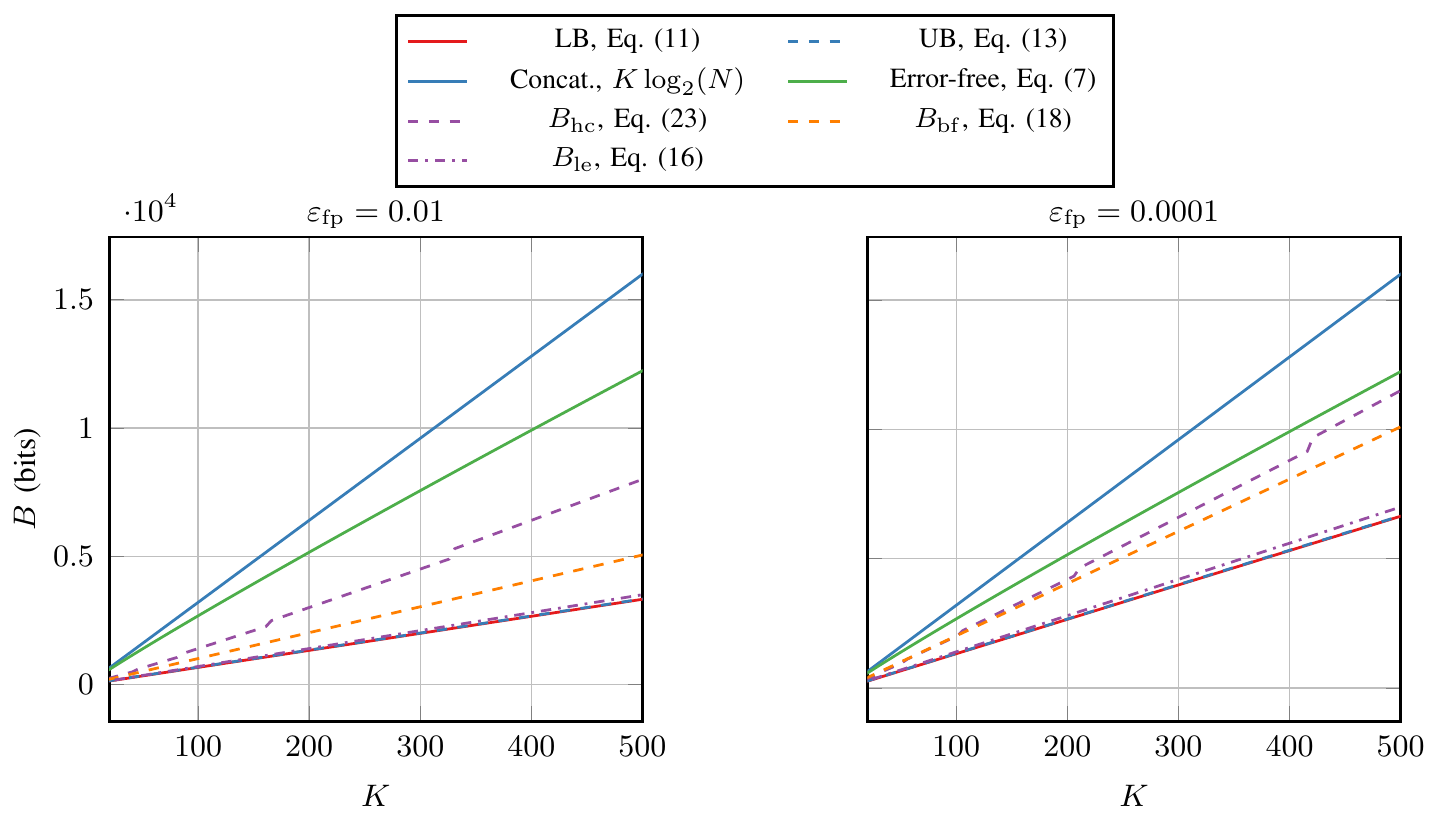}
    \caption{Message length, $B$, required to provide acknowledgment feedback for $N=2^{32}$ and $\epsilon_{\mathrm{fn}}=0$ with $\epsilon_{\mathrm{fp}}=0.01$ and $\epsilon_{\mathrm{fp}}=0.0001$.}
    \label{fig:bounds}
\end{figure}

\section{Random User Activity}\label{sec:random_activations}
So far, we have assumed that $K$ is fixed and optimized the feedback for a specific value of $K$. In practice, the number of active devices is random and unknown to both the BS and the devices, and thus the number of messages produced by the random access decoder at the BS, $K$, is in general also random. Furthermore, in the case where the users that were not decoded retransmit, $K$ may even be correlated over time and depend on the reliability of the feedback message itself. However, to simplify the analysis, we will here assume that $K$ is independent across frames and drawn from the distribution $p(K)$.

Construction and decoding of optimal feedback scheme based on linear equations require the number of users $K$ to be specified. On the other hand, the Bloom filter and hash concatenation schemes can be encoded and decoded without knowledge of the instantaneous $K$, but at the cost of a larger fraction of false positives. In this section, our focus is on the optimal scheme, and hence, we assume that the instantaneous value of $K$ is included in the feedback message, which incurs only a very small overhead.
To illustrate, suppose the random access mechanism is designed to support at most $K'=1024$ simultaneously active users, then a $10$-bit overhead is required to encode $K$. On the other hand, if the desired false positive probability is $\epsilon_{\mathrm{fp}}=0.001$, then approximately $\log_2(0.001)\approx 10$ bits are required per user in the feedback message, so the overhead introduced by encoding $K$ merely corresponds to encoding one additional user. When more than $K'$ users are active, we accept that the error probability can be arbitrarily high. In that case, we may decide to pick a random subset comprising $K'$ of the $K>K'$ decoded users at the cost of $K-K'$ false negatives.

If we allow the feedback message to have a variable length, then we can achieve the desired $\epsilon_{\mathrm{fp}}$ (and $\epsilon_{\mathrm{fn}}$) as long as $K\le K'$ without significant over-provisioning when $K<K'$ by adjusting the message length to $K$. %
However, the random user activity has a more significant impact on the performance when the length of the feedback message needs to remain fixed for any value of $K$, e.g., due to protocol constraints, as the error probabilities depends on the instantaneous value of $K$. It seems reasonable in this case to optimize message length either based on the average false positive/negative probabilities or by the probability that the these probabilities exceed some thresholds $\tilde{\epsilon}_{\mathrm{fp}}$ and $\tilde{\epsilon}_{\mathrm{fn}}$. Assuming for clarity that $\epsilon_{\mathrm{fn}}=0$, we can formalize the first case by defining the message length selection rule
\begin{align}
  B=\inf\left\{B'\ge 0: \E_{K\sim p(K)}[\epsilon_{\mathrm{fp}}(K,B')]\le \tilde{\epsilon}_{\mathrm{fp}}\right\},\label{eq:B_selection_expected}
\end{align}
where $\epsilon_{\mathrm{fp}}(K,B')$ is the false positive probability achieved with $K$ users and a message length of $B'$ bits, and $\tilde{\epsilon}_{\mathrm{fp}}$ is the specified false positive probability target. Similarly, for the second case we have
\begin{align}
  B=\inf\left\{B'\ge 0: \Pr(\epsilon_{\mathrm{fp}}(K,B')> \tilde{\epsilon}_{\mathrm{fp}}) \le \delta\right\},\label{eq:B_selection_exceed}
\end{align}
where $\delta$ specifies the maximum allowed probability that the false positive probability exceeds $\tilde{\epsilon}_{\mathrm{fp}}$.

Computing these feedback message lengths requires complete knowledge of the distribution of $K$, which is often not available. Instead, we proceed by deriving bounds based on the first moments of $p(K)$ using for simplicity the asymptotic expression for the feedback message length given in \cref{eq:asymptotic_bound} for $\epsilon_{\mathrm{fn}}=0$, which is accurate for large $N$. We note, however, that bounds can be extended to any of the practical schemes by first bounding the rounding error. For instance, for the scheme based on linear equations, we have $B_{\mathrm{le}}=K\lceil\log_2(1/\epsilon_{\mathrm{fp}})\rceil\le K(\log_2(1/\epsilon_{\mathrm{fp}})+1)$, which is straightforward to bound using the same methodology used in the following bounds. We first present an upper bound on the expected false positive $\bar{\epsilon}_{\mathrm{fp}}=\E_{K\sim p(K)}[2^{-B/K}]$.

\begin{proposition}\label{prop:expected_bound}
  Let the number of decoded users $K$ be random with mean $\bar{K}=\E[K]$ and variance $\Var[K]$. Then for a given message length $B$ the expected false positive probability $\bar{\epsilon}_{\mathrm{fp}}$ is upper bounded as
  \begin{align}
    \bar{\epsilon}_{\mathrm{fp}}< 2^{-B/\bar{K}} + \frac{2.66\Var[K]}{B^2}.\label{eq:expected_bound}
  \end{align}
\end{proposition}
\begin{proof}
By rearranging the expression in \cref{eq:asymptotic_bound} we obtain $\epsilon_{\mathrm{fp}}(K,B)\approx 2^{-B/K}$. To bound $\bar{\epsilon}_{\mathrm{fp}}$, we consider the first-order Taylor expansion of $2^{-B/K}$ around $\bar{K}=\E[K]$ given as
\begin{align}
  \begin{split}
   2^{-B/K}&=2^{-B/\bar{K}} + \frac{2^{-B/\bar{K}}}{\bar{K}^2}(K-\bar{K})\\
   &\quad+\frac{2^{-B/Z}B\ln(2)\left(B\ln(2)-2Z\right)}{Z^4}\frac{(K-\bar{K})^2}{2},
 \end{split}\label{eq:taylorapprox}
\end{align}
for some $Z$ between $\bar{K}$ and $K$. By analyzing its derivatives it can be shown that the term $\frac{2^{-B/Z}B\ln(2)\left(B\ln(2)-2Z\right)}{Z^4}$ is bounded and attains its maximum at $Z=\frac{\ln(2)(3-\sqrt{3})}{6}B$. From this we obtain the bound
\begin{align}
  \begin{split}
    \frac{2^{-B/Z}\left(B\ln(2)-2Z\right)}{Z^4} &\le \frac{e^{-\frac{6}{3-\sqrt{3}}} \left(\frac{6}{3-\sqrt{3}}-2\right)}{\left(\frac{3-\sqrt{3}}{6}\right)^3\ln^2(2)B^2}
  \end{split}\\
  &= \frac{\zeta}{B^2},
\end{align}
where $\zeta=e^{-\frac{6}{3-\sqrt{3}}}\left(\frac{6}{3-\sqrt{3}}-2\right)\left(\frac{6}{3-\sqrt{3}}\right)^{3}\ln^{-2}(2)$. By inserting into \cref{eq:taylorapprox}, taking expectation and rearranging we obtain
\begin{align}
  \bar{\epsilon}_{\mathrm{fp}}&\le 2^{-B/\bar{K}} + \frac{\zeta\Var[K]}{2B^2}.
\end{align}
The proof is completed by noting that $\zeta/2<2.66$.
\end{proof}

The result in \cref{prop:expected_bound} can be used to select the feedback message length according to the rule in \cref{eq:B_selection_expected}. We now derive a similar general bound on the probability that $\epsilon_{\mathrm{fp}}$ exceeds $\tilde{\epsilon}_{\mathrm{fp}}$ that can be used for the alternative feedback message length selection rule in \cref{eq:B_selection_exceed}.
\begin{proposition}\label{prop:exceed_bound}
  Let the number of decoded users $K$ be random with mean $\bar{K}=\E[K]$ and variance $\Var[K]$. For a given message length $B$ and $\tilde{\epsilon}_{\mathrm{fp}}\ge 2^{-B/\bar{K}}$, the probability that the false positive probability $\epsilon_{\mathrm{fp}}$ exceeds $\tilde{\epsilon}_{\mathrm{fp}}$ can be bounded as
\begin{align}
  \Pr\left(\epsilon_{\mathrm{fp}} > \tilde{\epsilon}_{\mathrm{fp}}\right) \le \frac{\Var[K]}{\left(\frac{B}{\log_2(1/\tilde{\epsilon}_{\mathrm{fp}})}-\bar{K}\right)^2}\label{eq:exceed_bound}
\end{align}
\end{proposition}
\begin{proof}
  Note first that
  \begin{align}
    \Pr\left(\epsilon_{\mathrm{fp}} > \tilde{\epsilon}_{\mathrm{fp}}\right) &= \Pr\left(K > \frac{B}{\log_2(1/\tilde{\epsilon}_{\mathrm{fp}})}\right)\\
     &\le \Pr\left(K \ge \frac{B}{\log_2(1/\tilde{\epsilon}_{\mathrm{fp}})}\right).
  \end{align}
  Applying Chebyshev's inequality yields
  \begin{align}
    \Pr\left(K \ge \frac{B}{\log_2(1/\tilde{\epsilon}_{\mathrm{fp}})}\right)&\le \Pr\left(|K-\bar{K}| \ge  \frac{B}{\log_2(1/\tilde{\epsilon}_{\mathrm{fp}})}-\bar{K}\right)\\
    &\le \frac{\Var[K]}{\left(\frac{B}{\log_2(1/\tilde{\epsilon}_{\mathrm{fp}})}-\bar{K}\right)^2}.
  \end{align}
\end{proof}

Because the bound in \cref{prop:exceed_bound} does not assume much about the distribution of $K$, it is in general not very tight. If we further assume that the users activate and are decoded independently (but not necessarily with the same probability), we can tighten the bound as presented in the following proposition. Because the number of decoded users are assumed to be independent, the bound depends only on the first moment of $K$.
\begin{proposition}\label{prop:exceed_bound_indep}
  Let the number of decoded users $K=\sum_{i=1}^N k_i$ where $k_i\in\{0,1\}$ are independently Bernoulli random variables with $\Pr(k_i=1)=p_i$, and let $\bar{K}=\E[K]=\sum_{i=1}^N p_i$. For a given feedback message length $B$ and $\tilde{\epsilon}_{\mathrm{fp}}\ge 2^{-B/\bar{K}}$, the probability that the false positive probability $\epsilon_{\mathrm{fp}}$ exceeds $\tilde{\epsilon}_{\mathrm{fp}}$ can be bounded as
  \begin{align}
    \Pr\left(\epsilon_{\mathrm{fp}} > \tilde{\epsilon}_{\mathrm{fp}}\right) &< \Pr\left(\frac{e^{\left(\eta_{\tilde{\epsilon}_{\mathrm{fp}}}-1\right)}}{\left(\eta_{\tilde{\epsilon}_{\mathrm{fp}}}\right)^{\eta_{\tilde{\epsilon}_{\mathrm{fp}}}}}\right)^{\bar{K}},\label{eq:exceed_bound_indep}
  \end{align}
  where $\eta_{\tilde{\epsilon}_{\mathrm{fp}}} = B/\left(\bar{K}\log_2(1/\tilde{\epsilon}_{\mathrm{fp}})\right)$.
\end{proposition}
\begin{proof}
  As in \cref{prop:exceed_bound} we have
  \begin{align}
    \Pr\left(\epsilon_{\mathrm{fp}} > \tilde{\epsilon}_{\mathrm{fp}}\right) &\le \Pr\left(K \ge \frac{B}{\log_2(1/\tilde{\epsilon}_{\mathrm{fp}})}\right).
  \end{align}
  Defining $\eta_{\tilde{\epsilon}_{\mathrm{fp}}} = B/\left(\bar{K}\log_2(1/\tilde{\epsilon}_{\mathrm{fp}})\right)$ and applying the Chernoff bound for Poisson trials (see e.g., Theorem 4.4 in~\cite{mitzenmacher2017probability}), we obtain
  \begin{align}
    \Pr\left(K \ge \frac{B}{\log_2(1/\tilde{\epsilon}_{\mathrm{fp}})}\right)&= \Pr\left(K \ge \eta_{\tilde{\epsilon}_{\mathrm{fp}}}\bar{K}\right)\\
    &< \Pr\left(\frac{e^{\left(\eta_{\tilde{\epsilon}_{\mathrm{fp}}}-1\right)}}{\left(\eta_{\tilde{\epsilon}_{\mathrm{fp}}}\right)^{\eta_{\tilde{\epsilon}_{\mathrm{fp}}}}}\right)^{\bar{K}},
  \end{align}
  which completes the proof.
\end{proof}

\section{Random Access with Feedback and Retransmissions}\label{sec:feedback_retx}
In this section, we analyze the impact of feedback in a scenario with $L$ transmission rounds, each comprising an uplink and a downlink phase. We first analyze the problem with packet erasure channels in both uplink and downlink, and then extend the analysis to a richer channel model that allows us to characterize the trade-off between false positives/negatives in the feedback message and the transmission rate. We assume that the transmission rounds are independent, and that the channel erasure probabilities are the same in all rounds.

\subsection{Packet Erasure Channels}
\begin{figure}
  \centering
\includegraphics{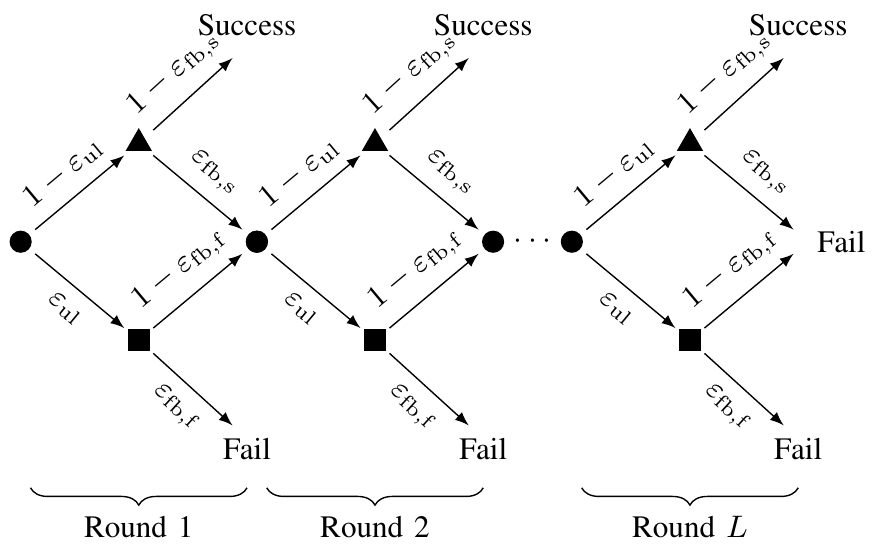}
  \caption{Events in the case with $L$ transmission rounds, where $\epsilon_{\mathrm{fb,s}}=1-(1-\epsilon_{\mathrm{dl}})(1-\epsilon_{\mathrm{fn}})$ and $\epsilon_{\mathrm{fb,f}}=(1-\epsilon_{\mathrm{dl}})\epsilon_{\mathrm{fp}}$. The circles represent the start of a round, triangles represent the case when a packet is successfully decoded by the BS, and the squares are when the packets are not decoded by the BS. A success occurs when a packet is both decoded by the BS and the user decodes the acknowledgment.}
  \label{fig:retransmission_model}
\end{figure}
We consider the transmission scenario from the perspective of a single user and assume an uplink erasure probability $\epsilon_{\mathrm{ul}}$, downlink erasure probability $\epsilon_{\mathrm{dl}}$, false positive probability $\epsilon_{\mathrm{fp}}$, and false negative probability $\epsilon_{\mathrm{fn}}$, which are the same in all transmission rounds. Under these conditions, the transmission process is illustrated in \cref{fig:retransmission_model}, where $\epsilon_{\mathrm{fb,s}}=1-(1-\epsilon_{\mathrm{dl}})(1-\epsilon_{\mathrm{fn}})$ and $\epsilon_{\mathrm{fb,f}}=(1-\epsilon_{\mathrm{dl}})\epsilon_{\mathrm{fp}}$ are the probabilities that the user makes a wrong decision based on the feedback, conditioned on success or failure in the uplink, respectively. We assume that the user succeeds only if it received an acknowledgment for the packet transmitted in the same round, i.e., a false positive acknowledgment is a failure even if the uplink was successful in a previous round but the downlink in that round was unsuccessful. Furthermore, the user retransmits if it is unable to decode the feedback. The failure probability is then given as
\begin{align}
  \Pr(\text{fail}) &= 1-\sum_{l=1}^{L} \ell^{l-1}(1-\epsilon_{\mathrm{ul}})(1-\epsilon_{\mathrm{fb,s}})\\
  &= 1-(1-\epsilon_{\mathrm{ul}})(1-\epsilon_{\mathrm{fb,s}})\left(\frac{1-\ell^L}{1-\ell}\right),\label{eq:pr_fail}
\end{align}
where $\ell=\epsilon_{\mathrm{ul}}(1-\epsilon_{\mathrm{fb,f}}) + (1-\epsilon_{\mathrm{ul}})\epsilon_{\mathrm{fb,s}}$ is the probability that the user proceeds from one transmission round to the next. To gain some insight into the behavior, suppose first that $L=1$, in which case the expression reduces to
\begin{align}
  \Pr(\text{fail}) &=1-(1-\epsilon_{\mathrm{ul}})(1-\epsilon_{\mathrm{dl}})(1-\epsilon_{\mathrm{fn}}),\label{eq:pr_fail_onetx}
\end{align}
suggesting that $\epsilon_{\mathrm{ul}}$, $\epsilon_{\mathrm{dl}}$ and $\epsilon_{\mathrm{fn}}$ have equal importance in minimizing the failure probability. Furthermore, because false positives can only occur when the uplink fails, in which case the entire transmission fails since $L=1$, the failure probability is independent of the false positive probability $\epsilon_{\mathrm{fp}}$. Similarly, suppose now that we allow an infinite number of retransmissions. By taking the limit $L\to \infty$ in \cref{eq:pr_fail} we obtain
\begin{align}
  \Pr(\text{fail}) &=1- \frac{(1-\epsilon_{\mathrm{ul}})(1-\epsilon_{\mathrm{fn}})}{1-\epsilon_{\mathrm{fn}}-\epsilon_{\mathrm{ul}}(1-\epsilon_{\mathrm{fn}}-\epsilon_{\mathrm{fp}})}.\label{eq:pr_fail_limit}
\end{align}
Note that this expression depends on the uplink erasure probability and the false positive/negative probabilities, but not on the downlink erasure probability $\epsilon_{\mathrm{dl}}$, which is because an erasure in the downlink always will result in a retransmission. When $\epsilon_{\mathrm{fp}}=0$ the expression reduces to $\Pr(\text{fail})=0$ indicating that even with false negatives (but no false positives), an arbitrarily high reliability can be achieved by increasing the transmission rounds. On the other hand, when $\epsilon_{\mathrm{fn}}= 0$ the expression in \cref{eq:pr_fail_limit} reduces to
\begin{align}
  \Pr(\text{fail}) &=1- \frac{1-\epsilon_{\mathrm{ul}}}{1-\epsilon_{\mathrm{ul}}(1-\epsilon_{\mathrm{fp}})},\label{eq:pr_fail_limit_nofn}
\end{align}
suggesting that when $\epsilon_{\mathrm{fp}}>0$ retransmissions cannot fully compensate for an unreliable uplink channel. The intuition behind this is that an unreliable uplink increases the probability of receiving a false positive, which in turn increases the failure probability.

Returning to the case with finite $L$, by rearranging \cref{eq:pr_fail} we can also obtain an expression for the number of transmission rounds $L$ required to achieve a target failure probability $\epsilon_{\mathrm{fail}}$ as
\begin{align}
  L &\ge \left\lceil\frac{\ln\left(1-\frac{1-\epsilon_{\mathrm{fail}}}{(1-\epsilon_{\mathrm{ul}})(1-\epsilon_{\mathrm{fb,s}})}(1-\ell)\right)}{\ln(\ell)}\right\rceil,
\end{align}
which is valid when $\epsilon_{\mathrm{fail}}$ is greater than the asymptotic failure probability given by \cref{eq:pr_fail_limit}. In particular, to be within a factor $(1+\rho)$ of the asymptotic error probability, one needs
\begin{align}
  L &\ge \left\lceil\frac{\ln\left(\frac{\rho\epsilon_{\mathrm{ul}}\epsilon_{\mathrm{fp}}}{(1-\epsilon_{\mathrm{fn}})(1-\epsilon_{\mathrm{ul}})}\right)}{\ln(\ell)}\right\rceil.
\end{align}

Note that the analysis above holds even when the number of users $K$ is random if the length of the feedback message and the transmission rate (channel coding rate) of the feedback message are adapted based on the instantaneous $K$ to match $\epsilon_{\mathrm{dl}}$, $\epsilon_{\mathrm{fp}}$, and $\epsilon_{\mathrm{fn}}$. However, if the length of the feedback message remains fixed, $\epsilon_{\mathrm{fp}}$ and $\epsilon_{\mathrm{fn}}$ depends on the instantaneous $K$. In this case, a reasonable strategy is to use \cref{eq:pr_fail} to determine an appropriate $\epsilon_{\mathrm{fp}}$, and then apply either \cref{prop:expected_bound} or \cref{prop:exceed_bound} to select the feedback message length such that the target false positive probability is satisfied with the desired probability.

\subsection{Source/Channel Coding Trade-off}\label{sec:src_ch_tradeoff}
In practice, the erasure probability of the downlink transmission, $\epsilon_{\mathrm{dl}}$, is a function of the transmission rate and depends on the SNR at the receiver. Furthermore, for a given number of symbols transmitted over the channel, the rate depends on the length of the feedback message $B$, which directly impacts the false positive/negative probabilities. Consequently, there is an inherent trade-off between $\epsilon_{\mathrm{dl}}$, $\epsilon_{\mathrm{fp}}$ and $\epsilon_{\mathrm{fn}}$, which determine the overall reliability of the system. To illustrate the trade-off, suppose we can construct a feedback message with $\epsilon_{\mathrm{fn}}=0$ and false positive probability $\epsilon_{\mathrm{fp}}$ using $K\log_2(1/\epsilon_{\mathrm{fp}})$ bits, and that we aim to transmit it over a quasi-static fading channel with additive noise and instantaneous SNR given by $\gamma$. For a given number of symbols $c$, the transmission rate is given as $K\log_2(1/\epsilon_{\mathrm{fp}})/c$ and the probability of decoding error is thus
\begin{align}
  \epsilon_{\mathrm{dl}} &= \Pr\left(\log_2(1+\gamma) < \frac{K\log_2(1/\epsilon_{\mathrm{fp}})}{c}\right)\\
  &=\Pr\left(\gamma < \epsilon_{\mathrm{fp}}^{-K/c}-1\right),\label{eq:outage_dl}
\end{align}
illustrating, as expected, that decreasing $\epsilon_{\mathrm{fp}}$ causes $\epsilon_{\mathrm{dl}}$ to increase since a higher transmission rate is required.

If the number of symbols for the feedback message and the number of retransmissions $L$ are fixed, the transmission rate and the false positive/negative probabilities need to be jointly optimized to minimize the failure probability in \cref{eq:pr_fail}. This is in general a non-convex optimization problem that requires numerical evaluation of \cref{eq:pr_fail} over a range of $\epsilon_{\mathrm{fp}}$. The asymptotic expression in \cref{eq:pr_fail_limit} suggests that when $L$ is large, we should aim at minimizing $\epsilon_{\mathrm{fp}}$ and $\epsilon_{\mathrm{fn}}$ since $\epsilon_{\mathrm{dl}}$ has no impact on the failure probability. In particular, in this case the introduction of false positives will lead to a worse performance compared to identifier concatenation, as there is no gain in reducing the length of the acknowledgment packet. However, when $L$ is small, the downlink erasure probability has an increasing impact since a successful downlink transmission is required to succeed. In particular, to minimize the failure probability for $L=1$ the downlink probability and the false negative probability should be equal, while $\epsilon_{\mathrm{fp}}$ has no impact, as can be seen in \cref{eq:pr_fail_onetx}.

\section{Numerical Results}\label{sec:num_res}
In this section, we evaluate the feedback schemes in a typical massive random access setting. We first present results that illustrate the impact of false positives in the setting with multiple transmission rounds and fixed $K$ over a simple erasure channel. We then investigate the case with random $K$, and finally we exemplify the trade-off between allocating channel symbols for the uplink and the feedback under a random access channel in the uplink and a Rayleigh fading channel in the downlink. Except for the cases where it is explicitly mentioned, we will assume that $\epsilon_{\mathrm{fn}}=0$.

\begin{figure}
  \centering
\includegraphics{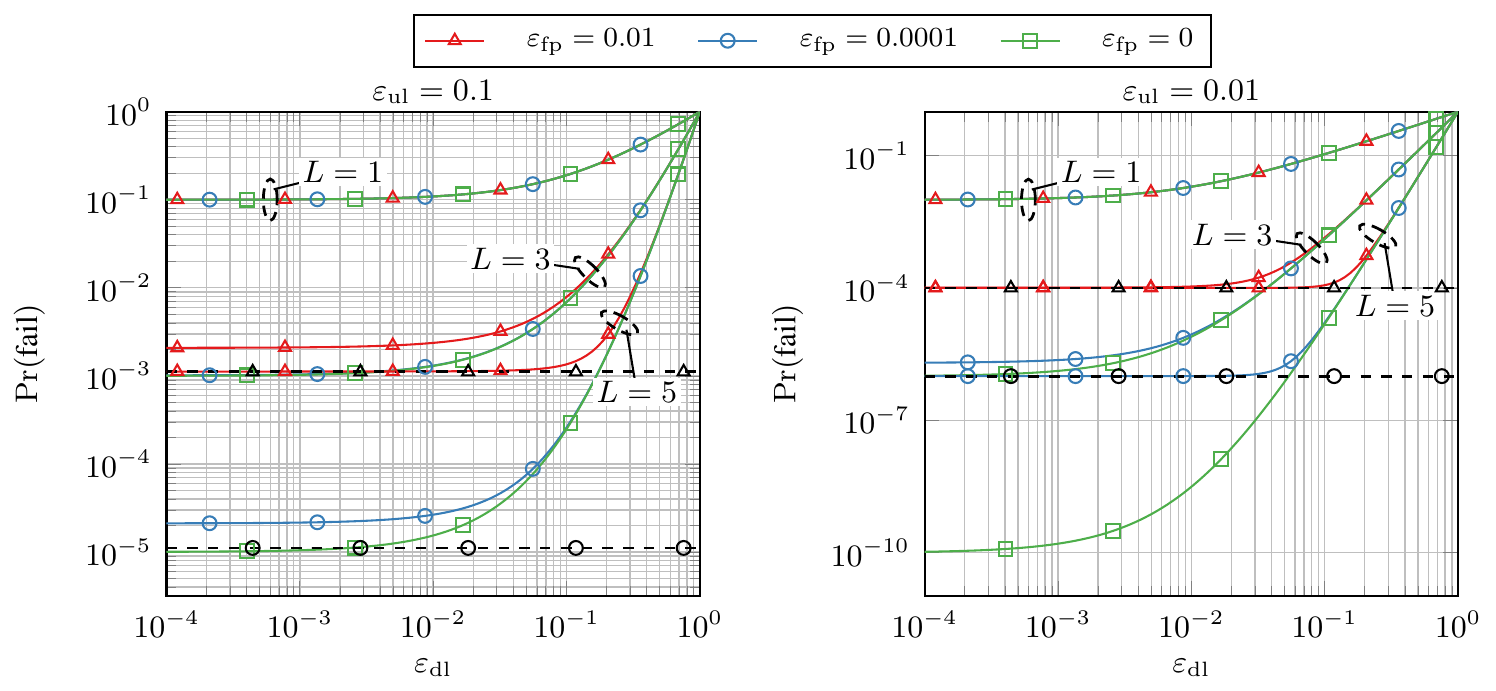}
  \caption{The probability that a transmission fails with $L$ rounds obtained using \cref{eq:pr_fail}. The dashed lines show the asymptotic results for $L\to\infty$.}
  \label{fig:retransmissions_pec}
\end{figure}

\subsection{Fixed $K$ and $L$ Transmission Rounds}
When $K$ is fixed, the false positive probability $\epsilon_{\mathrm{fp}}$ is constant and can be picked arbitrarily by choosing an appropriate feedback message length $B$. The probability that a transmission fails in a setting with $L$ retransmissions is shown in \cref{fig:retransmissions_pec} along with the asymptotic results for $L\to\infty$, obtained using \cref{eq:pr_fail,eq:pr_fail_limit_nofn}, respectively. Although the downlink erasure probability $\epsilon_{\mathrm{dl}}$ has no impact as $L\to\infty$, it has a significant impact when $L$ is small. In particular, for finite $L$ the failure probability is at least $(\epsilon_{\mathrm{dl}})^L$, as a successful downlink transmission is required in order for the user to succeed. Similarly, we can observe an error floor as $\epsilon_{\mathrm{dl}}$ approaches zero caused by both the false positive probability and the uplink erasure probability. When $\epsilon_{\mathrm{fp}}$ is small the floor is approximately at $(\epsilon_{\mathrm{ul}})^L$. On the other hand, when $\epsilon_{\mathrm{ul}}$ is small, the error floor is dominated by $\epsilon_{\mathrm{fp}}$.

For low $\epsilon_{\mathrm{dl}}$, the failure probabilities are rather close to the asymptotic failure probabilities despite $L$ being as low as $3$ or $5$ (where the solid and dashed lines coincide). In this regime, the failure probability is limited only by $\epsilon_{\mathrm{fp}}$ and $\epsilon_{\mathrm{ul}}$, and increasing the downlink reliability or the number of transmission rounds will not lead to a reduced failure probability.

\subsection{Random $K$ and $L$ Transmission Rounds}
We now study the case when $K$ is random, and start by assessing the accuracy of the bounds derived in \cref{prop:expected_bound,prop:exceed_bound,prop:exceed_bound_indep} and investigating how the false positive probability depends on the distribution of $K$ when the length of the feedback message remains fixed to $B$ bits.
The impact of the distribution of $K$ is illustrated in \cref{fig:poisson_activations}, where $K$ follows a Poisson distribution with mean $\lambda$. \cref{fig:poisson_activations_expected} shows the expected false positive probability $\bar{\epsilon}_{\mathrm{fp}}$, computed numerically, and the bound from \cref{prop:expected_bound} when the feedback message length $B$ is optimized to provide a false positive probability $\tilde{\epsilon}_{\mathrm{fp}}=0.0001$ when $K=K'$. As can be seen, the expected false positive probability is larger than the target false positive probability of $\tilde{\epsilon}_{\mathrm{fp}}=0.0001$ when $\lambda=K'$, suggesting that optimizing based on only the expected number of active users is insufficient. However, the bound, which also takes into account the variance of $K$, is accurate when $\lambda$ is close to and greater than $K'$, and can be used to pick a feedback message length that satisfies the target false positive probability when $\lambda=K'$ at the cost of only a minor message length penalty.

The probability that $\epsilon_{\mathrm{fp}}$ exceeds $\tilde{\epsilon}_{\mathrm{fp}}=0.0001$ is shown in \cref{fig:poisson_activations_exceed} along with the bounds from \cref{prop:exceed_bound,prop:exceed_bound_indep}. While the bound from \cref{prop:exceed_bound} is reasonable when $\lambda$ is close to $K'$, it is generally quite weak due to the strong concentration of the Poisson distribution around its mean. However, by assuming that the users activate independently as in \cref{prop:exceed_bound_indep} the bound can be significantly tightened especially for low $\lambda$.

\begin{figure}
  \centering
  \ifdefined\SINGLECOL\begin{minipage}{0.5\textwidth}\fi
  \centering
  \begin{subfigure}[c]{0.95\linewidth}
    \centering
\includegraphics{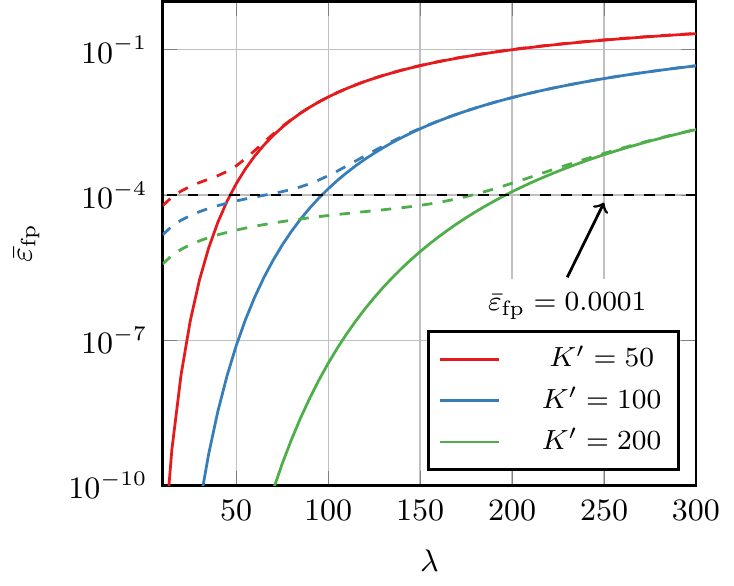}
    \caption{}
    \label{fig:poisson_activations_expected}
  \end{subfigure}
  \ifdefined\SINGLECOL\end{minipage}\hfill\begin{minipage}{0.5\textwidth}\fi
  \begin{subfigure}[c]{0.95\linewidth}
    \centering
\includegraphics{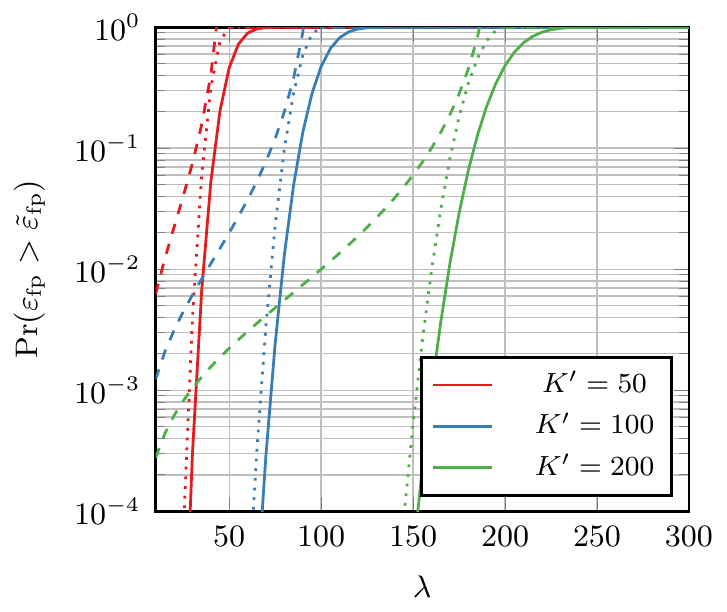}
    \caption{}
    \label{fig:poisson_activations_exceed}
  \end{subfigure}
  \ifdefined\SINGLECOL\end{minipage}\fi
  \caption{Illustrations of the bounds for (\subref{fig:poisson_activations_expected}) the expected false positive probability $\bar{\epsilon}_{\mathrm{fp}}$, and (\subref{fig:poisson_activations_exceed}) the probability that $\epsilon_{\mathrm{fp}}$ exceeds $\tilde{\epsilon}_{\mathrm{fp}}=0.0001$ when $K$ is Poisson distributed with mean $\lambda$. The feedback message length $B$ is optimized to guarantee a false positive probability of $\tilde{\epsilon}_{\mathrm{fp}}=0.0001$ when $K=K'$. The solid lines indicate the actual probabilities obtained using the full distribution, and in (\subref{fig:poisson_activations_expected}) the dashed line shows \cref{prop:expected_bound}, while in (\subref{fig:poisson_activations_exceed}) the dashed and dotted lines show \cref{prop:exceed_bound} and \cref{prop:exceed_bound_indep}, respectively.}
  \label{fig:poisson_activations}
\end{figure}

We now turn our attention to the case with $L$ transmission rounds, and assume that the length of the feedback message $B$ is selected using \cref{prop:expected_bound} to satisfy a given false positive requirement $\tilde{\epsilon}_{\mathrm{fp}}$ on average. The failure probability when $K$ is Poisson distributed is shown in \cref{fig:retransmissions_pec_poisson} for $L=5$ and $\epsilon_{\mathrm{dl}}=0.01$. Because the message length is selected using the bound from \cref{prop:expected_bound}, the failure probability for random $K$ is lower than the one with deterministic $K$, indicated by the dashed lines. The gap between the failure probability for deterministic $K$ and random $K$ decreases as $\lambda$ increases, which is due to the bound becoming more tight in this regime. This confirms that the bound can be used as a useful tool to select the message length.

\begin{figure}
  \centering
\includegraphics{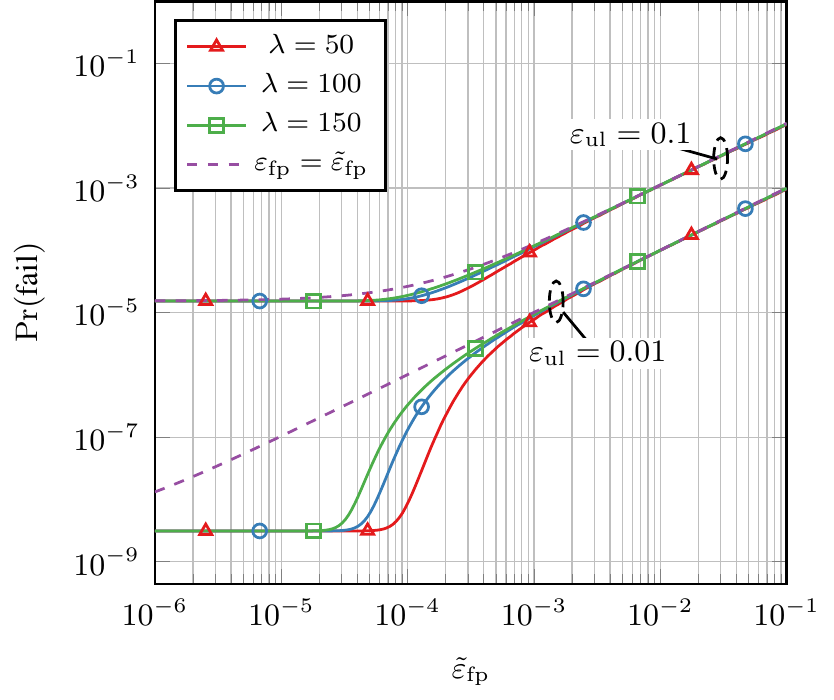}
  \caption{Failure probability vs. target average false positive probability $\tilde{\epsilon}_{\mathrm{fp}}$ for Poisson arrivals with mean $\lambda$, for $L=5$ and $\epsilon_{\mathrm{dl}}=0.01$. The message lengths are selected using the bound in \cref{prop:expected_bound}, and the dashed lines show the failure probability when $K$ is deterministic.}
  \label{fig:retransmissions_pec_poisson}
\end{figure}

\begin{figure}
  \centering
\includegraphics{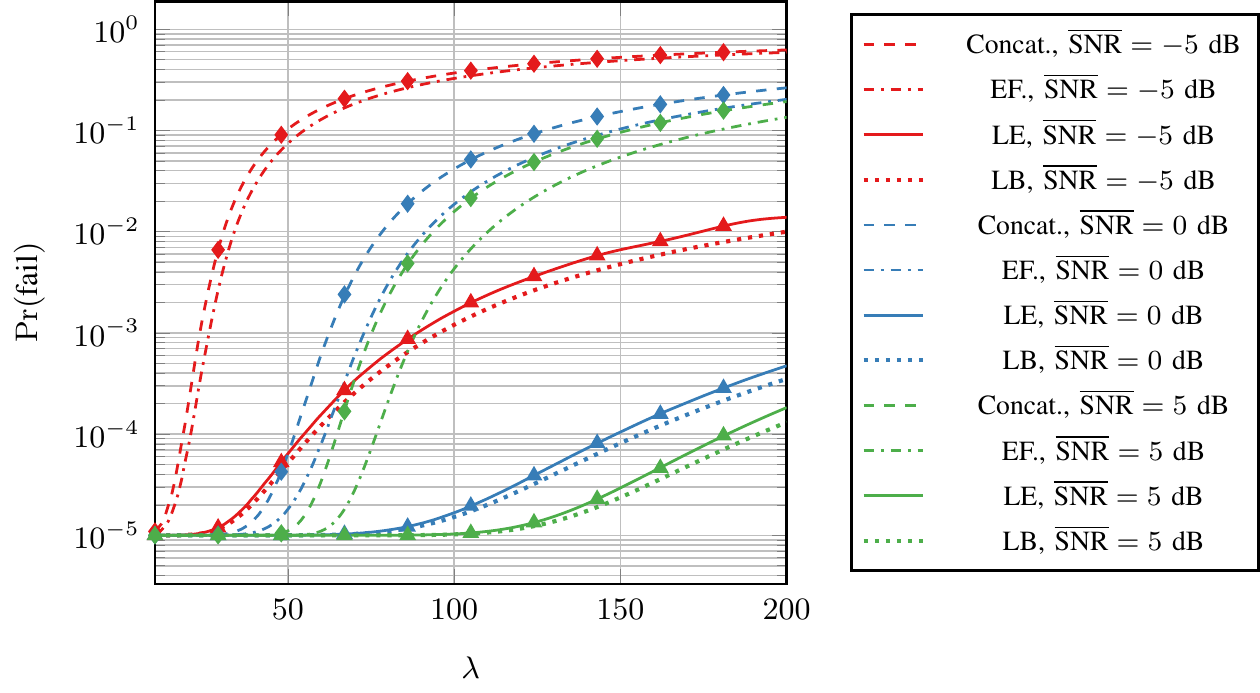}
  \caption{Failure probability for the concatenation based encoding compared to the error-free (EF) bound (\cref{eq:bound_err_free}), the linear equations (LE) scheme (\cref{eq:r_le}), and the lower bound (LB, \cref{eq:asymptotic_bound}) when the acknowledgment is transmitted over a Rayleigh fading channel with $c=2048$ symbols and $64$ transmitter antennas. $K$ is Poisson distributed with mean $\lambda$, $\epsilon_{\mathrm{ul}}=0.1$ and $L=5$. Diamond and triangle markers are obtained by simulation of the concatenation and LE schemes, respectively.}
  \label{fig:joint_src_ch_coding}
\end{figure}

\subsection{Source/Channel Coding Trade-off}
We finish the section by studying the trade-off between the number of bits used to encode the acknowledgments and the transmission rate. We assume that $K$ is Poisson distributed with arrival rate $\lambda$ and the uplink reliability is $\epsilon_{\mathrm{ul}}=0.1$. For the downlink, we pick $\epsilon_{\mathrm{dl}}$ using \cref{eq:outage_dl} for the case in which the BS has $64$ antennas, there are $L=5$ transmission rounds, and $c=2048$ channel symbols are available for the feedback, such that the transmission rate is $B/2048$ bits/symbol. Assuming a quasi-static flat-fading Rayleigh channel with average SNR $\overline{\text{SNR}}$, the instantaneous SNR at the user, $\gamma$, is Gamma distributed with shape and scale parameters equal to $64$ and $\overline{\text{SNR}}/64$, respectively. We consider four encoding methods, namely identifier concatenation, the error-free (EF) method from \cref{eq:bound_err_free}, the scheme based on linear equations (LE) presented in \cref{sec:le}, and the asymptotically optimal scheme from \cref{eq:asymptotic_bound}. For each value of $\lambda$ and each encoding scheme, we optimize $B$ so that $\Pr(\text{fail})$ is minimized when averaged over the instantaneous arrivals given $\lambda$. Thus, the transmission rate remains fixed for a given $\lambda$. In the concatenation and error-free schemes, we assume that each identifier requires $32$ bits, and when the length of the feedback message is less than $32K$ (if the instantaneous $K$ is large compared to the message length), we encode a random subset comprising $\lceil B/32\rceil$ identifiers, which results in a false negative probability of $\epsilon_{\mathrm{fn}}=1-\lceil B/32\rceil/K$ (but no false positives). Therefore, while these representations are error-free when the number of bits $B$ is adapted to $K$, they are \emph{not} error free here where $K$ is random and the number of bits is optimized to minimize the failure probability.

The results are shown in \cref{fig:joint_src_ch_coding} for $\overline{\text{SNR}}\in\{-5,0,5\}~\text{dB}$. The figure shows that, despite introducing false positives, the failure probability can be substantially decreased when the scheme based on linear equations is used compared to both the straightforward concatenation scheme and the bound given by the EF scheme. This is because admitting false positives allows the message length to be significantly reduced (and thus, the transmission rate), which in turn leads to much higher reliability of the downlink feedback. Furthermore, as expected the scheme based on linear equations performs close to the asymptotically optimal bound, with a gap caused only by the rounding in \cref{eq:r_le}. Finally, we see that simulations, indicated by the markers, agree with the theoretical analysis, manifesting that the gains can be attained in practice.

The failure probability vs. average SNR is shown for the same scenario in \cref{fig:joint_src_ch_coding_vs_snr} for $\lambda=100$, illustrating that the method based on linear equations is superior for a wide range of SNRs, especially in the low-SNR regime where the compressed representation allows the message to be transmitted at a lower rate. The figure reveals that in order for the error-free methods to perform better, the SNR must be high and the number of transmission rounds must be large, so that neither the downlink rate nor the uplink success probability are limiting the performance, but only the false positives. In particular, by comparing \cref{fig:joint_src_ch_coding_vs_snr_01} and \cref{fig:joint_src_ch_coding_vs_snr_001} for $L=10$, it is clear that the error-free methods are more beneficial in the latter case where the uplink success probability is higher and thus not limiting the performance as much.
\begin{figure}
  \centering
  \ifdefined\SINGLECOL\begin{minipage}{0.5\textwidth}\fi
  \centering
  \begin{subfigure}[c]{0.95\linewidth}
    \centering
\includegraphics{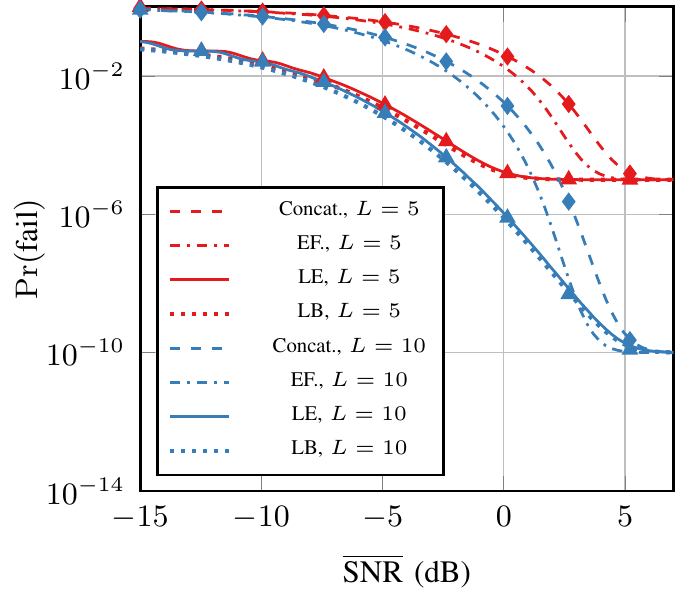}
    \caption{$\epsilon_{\mathrm{ul}}=0.1$.}
    \label{fig:joint_src_ch_coding_vs_snr_01}
  \end{subfigure}
  \ifdefined\SINGLECOL\end{minipage}\hfill\begin{minipage}{0.5\textwidth}\fi
  \begin{subfigure}[c]{0.95\linewidth}
    \centering
\includegraphics{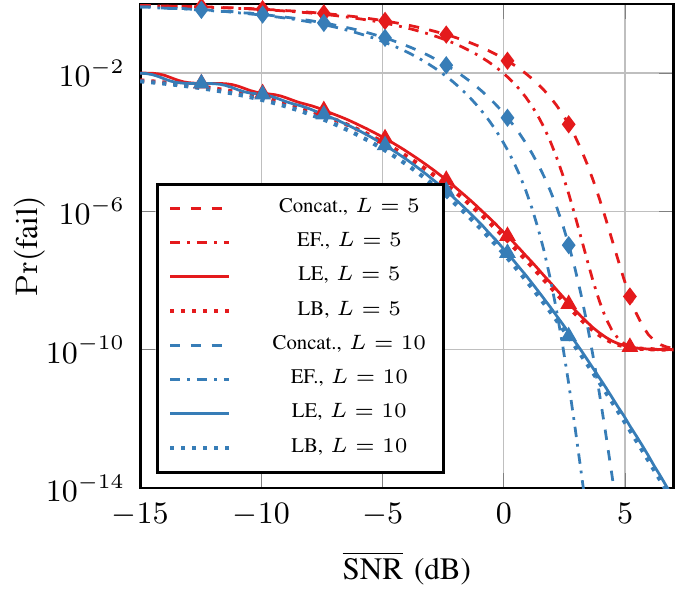}
    \caption{$\epsilon_{\mathrm{ul}}=0.01$.}
    \label{fig:joint_src_ch_coding_vs_snr_001}
  \end{subfigure}
  \ifdefined\SINGLECOL\end{minipage}\fi
  \caption{Failure probability vs. SNRs for $L$ transmission rounds and $\lambda=100$. The remaining parameters are the same as in \cref{fig:joint_src_ch_coding}.}
  \label{fig:joint_src_ch_coding_vs_snr}
\end{figure}

\begin{figure}
  \centering
\includegraphics{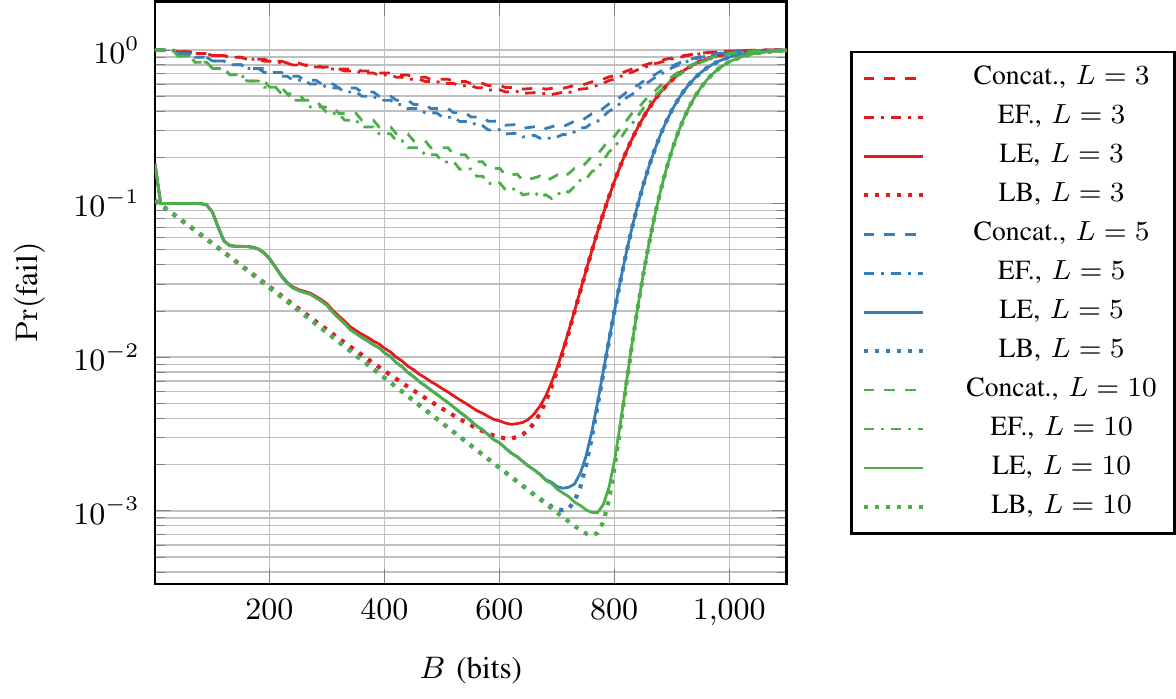}
  \caption{Failure probability for various acknowledgment message lengths $B$ for $\lambda=100$, $\epsilon_{\mathrm{ul}}=0.1$ and $\overline{\text{SNR}}=-5$ dB. The remaining parameters are the same as in \cref{fig:joint_src_ch_coding}.}
  \label{fig:joint_src_ch_coding_vs_n}
\end{figure}

\cref{fig:joint_src_ch_coding_vs_n} shows the failure probability vs. feedback message length $B$ for $\lambda=100$, $\overline{\text{SNR}}=-5$ dB, and various number of transmission rounds $L$. In general, the figure reveals the trade-off between source and channel coding. When $B$ is small, the feedback message is transmitted at a low rate and is likely to be decoded, but has a high fraction of false positives, causing a high failure probability. Conversely, when $B$ is large the message has a small fraction of false positives, but it is transmitted at a high rate which reduces the probability that it is decoded, which again leads to a high failure probability. Therefore, the optimal $B$ is obtained by balancing the message length and the transmission rate. As also suggested by the analysis in \cref{sec:src_ch_tradeoff}, for the asymptotic lower bound and the scheme based on linear equations, the feedback message length $B$ that minimizes the failure probability increases as the number of transmission rounds increases, causing $\epsilon_{\mathrm{dl}}$ to decrease. On the other hand, for the concatenation scheme, the feedback message length that minimizes the failure probability is the same independently of $L$. This is because $\epsilon_{\mathrm{fn}}$, which increases with $B$, and $\epsilon_{\mathrm{dl}}$, which decreases with $B$, both lead to the same event, namely a retransmission. This point has a high false negative probability of $\epsilon_{\mathrm{fn}}\approx 0.78$, while the outage probability is low ($\epsilon_{\mathrm{dl}}\approx 0.05$). On the other hand, the schemes that allows for false positives has $\epsilon_{\mathrm{fp}}$ in the range $0.005$ to $0.021$, while $\epsilon_{\mathrm{dl}}$ is ranges from approximately $0.01$ to $0.30$. This illustrates, in line with existing literature, the fact that the false positive probability should generally be kept smaller then the false negative probability. Despite this, the resulting failure probability is significantly smaller when false positives are allowed, compared to the case where they are not.

As a final remark, we note that the results in this section have been obtained under the assumption that the BS does not have channel state information (CSI) of the decoded users, although many massive random access schemes obtain this as part of the decoding procedure~\cite{liu18,fengler21}. When CSI is available, the BS can increase the SNR at the devices that were successful in the uplink, while the unsucessful users will experience a lower SNR. This effectively suppresses the false positive probability, leading to an even smaller failure probability.

\section{Conclusion}\label{sec:conclusion}
In this work, we have studied the use of message acknowledgments in a massive random access setting. We have shown that because of the large number of users that are active at any given time, encoding the feedback message requires a significant number of bits. To reduce this amount, we propose to allow for a small fraction of false positive acknowledgments, which results in a significant reduction in the length of the acknowledgment message. We have presented and analyzed a number of practical schemes of various complexity that can be used to realize these reductions, and shown that their performance is close to the information-theoretic optimum. With the basis of these schemes, we have studied their performance when the number of decoded users is random, and derived bounds on the false positive probability in this setting. Finally, we have studied how the schemes perform in a scenario with retransmissions, and shown, through numerical results, the extent to which reducing the feedback message length can improve the overall reliability of the random access scenarios.

\appendices
\crefalias{section}{appendix}
\section{Derivation of \cref{eq:lower}}\label{app:inf_lowerbound}
For completeness, we derive here the lower bound in \cref{eq:lower} presented as Proposition 4 in \cite{pagh01}.

Suppose we construct a feedback message that acknowledges a set of users $\mathcal{W}\subset [N]$. We are interested in finding the number of sets $\mathcal{S}$ of size $K$ that such a message can acknowledge while the requirements in terms of false positives and false negatives are satisfied. Clearly, in order to meet the false positive requirement, we must have $|\mathcal{W}|\le K+\lfloor\epsilon_{\mathrm{fp}}N\rfloor$.

Consider first the sets $\mathcal{S}$ for which the message $\mathcal{W}$ has exactly $i$ false negatives and thus $K-i$ true positives. In order for $\mathcal{W}$ to be a valid message for such a set, at least $K-i$ users of $\mathcal{S}$ must belong to $\mathcal{W}$, while the remaining $i$ users can be any of the $N-|\mathcal{W}|$ users that are not acknowledged by $\mathcal{W}$. For a given message $\mathcal{W}$, the number of such sets is $\binom{|\mathcal{W}|}{K-i}\binom{N-|\mathcal{W}|}{i}\le \binom{K+\lfloor\epsilon_{\mathrm{fp}}N\rfloor}{K-i}\binom{N}{i}$.
Thus, the number of sets with up to $\lfloor\epsilon_{\mathrm{fn}}K\rfloor$ false negatives is at most
\ifdefined\SINGLECOL
\begin{equation}
  \sum_{i=0}^{\lfloor\epsilon_{\mathrm{fn}}K\rfloor} \binom{K+\lfloor\epsilon_{\mathrm{fp}}N\rfloor}{K-i}\binom{N}{i} \le  K\binom{K+\lfloor\epsilon_{\mathrm{fp}}N\rfloor}{K-\lfloor\epsilon_{\mathrm{fn}}K\rfloor}\binom{N}{\lfloor\epsilon_{\mathrm{fn}}K\rfloor}.
\end{equation}
\else
\begin{multline}
  \sum_{i=0}^{\lfloor\epsilon_{\mathrm{fn}}K\rfloor} \binom{K+\lfloor\epsilon_{\mathrm{fp}}N\rfloor}{K-i}\binom{N}{i} \le\\
  K\binom{K+\lfloor\epsilon_{\mathrm{fp}}N\rfloor}{K-\lfloor\epsilon_{\mathrm{fn}}K\rfloor}\binom{N}{\lfloor\epsilon_{\mathrm{fn}}K\rfloor}.
\end{multline}
\fi
The total number of bits to represent all $\binom{N}{K}$ possible sets $\mathcal{S}$ is therefore at most
\begin{align}
  B_{\text{fp,fn}}^* &\ge  \log_2\left(\frac{\binom{N}{K}}{K\binom{K+\lfloor\epsilon_{\mathrm{fp}}N\rfloor}{K-\lfloor\epsilon_{\mathrm{fn}}K\rfloor}\binom{N}{\lfloor\epsilon_{\mathrm{fn}}K\rfloor}}\right)\\
  &= \log_2\binom{N}{K}-\log_2\left(K\binom{K+\lfloor\epsilon_{\mathrm{fp}}N\rfloor}{\lceil (1-\epsilon_{\mathrm{fn}})K\rceil}\binom{N}{\lfloor\epsilon_{\mathrm{fn}}K\rfloor}\right).
\end{align}
Note that this bound is valid only when $\epsilon_{\mathrm{fp}}<1/2$, as it otherwise might be beneficial to encode the users that should \emph{not} be acknowledged instead of the users that should.

\bibliographystyle{IEEEtran}
\bibliography{bibliography}
\end{document}